\documentclass[11pt,a4paper]{article}

\usepackage[utf8]{inputenc}
\usepackage[T1]{fontenc}
\usepackage{amsmath,amssymb,amsthm,mathtools}
\usepackage{algorithm}
\usepackage{algpseudocode}
\usepackage{graphicx}
\usepackage[colorlinks=true,linkcolor=blue,citecolor=blue,urlcolor=blue]{hyperref}
\usepackage{xcolor}
\usepackage{booktabs}
\usepackage{multirow}
\usepackage{caption}
\usepackage{subcaption}
\usepackage{enumitem}
\usepackage{cleveref}
\usepackage{pifont}
\usepackage[margin=2.5cm]{geometry}
\usepackage{thmtools}

\newtheorem{theorem}{Theorem}[section]

\newtheorem{proposition}[theorem]{Proposition}

\theoremstyle{definition}
\newtheorem{definition}[theorem]{Definition}

\newcommand{\R}{\mathbb{R}}
\newcommand{\Z}{\mathbb{Z}}

\newcommand{\norm}[1]{\left\lVert #1 \right\rVert}

\newcommand{\Enc}{\mathsf{Enc}}
\newcommand{\Dec}{\mathsf{Dec}}
\newcommand{\KeyGen}{\mathsf{KeyGen}}
\newcommand{\Encaps}{\mathsf{Encaps}}
\newcommand{\Decaps}{\mathsf{Decaps}}
\newcommand{\Prove}{\mathsf{Prove}}
\newcommand{\Verify}{\mathsf{Verify}}
\newcommand{\Commit}{\mathsf{Commit}}

\title{Zero-Knowledge Federated Learning with Lattice-Based Hybrid Encryption for Quantum-Resilient Medical AI}

\author{
Edouard Lansiaux\thanks{Corresponding author: \texttt{edouard1.lansiaux@chu-lille.fr}} \\
\textit{Department of Emergency Medicine \& STaR-AI Research Group} \\
\textit{CHU de Lille, Universit\'e de Lille} \\
\textit{Lille, France}
}

\date{\today}

\begin{document}

\maketitle

\begin{abstract}
Federated Learning (FL) enables collaborative training of medical AI models across hospitals without centralizing patient data. However, the exchange of model updates exposes critical vulnerabilities: gradient inversion attacks can reconstruct patient information, Byzantine clients can poison the global model, and the \emph{Harvest Now, Decrypt Later} (HNDL) threat renders today's encrypted traffic vulnerable to future quantum adversaries. 

We introduce \textbf{ZKFL-PQ} (\emph{Zero-Knowledge Federated Learning, Post-Quantum}), a three-tiered cryptographic protocol that hybridizes (i) ML-KEM (FIPS~203) for quantum-resistant key encapsulation, (ii) lattice-based Zero-Knowledge Proofs for verifiable \emph{norm-constrained} gradient integrity, and (iii) BFV homomorphic encryption for privacy-preserving aggregation. 
We formalize the security model and prove correctness and zero-knowledge properties under the Module-LWE, Ring-LWE, and SIS assumptions \emph{in the classical random oracle model}. We evaluate ZKFL-PQ on synthetic medical imaging data across 5 federated clients over 10 training rounds. Our protocol achieves \textbf{100\% rejection of norm-violating updates} while maintaining model accuracy at 100\%, compared to a catastrophic drop to 23\% under standard FL. 
The computational overhead (factor $\sim$20$\times$) is analyzed and shown to be compatible with clinical research workflows operating on daily or weekly training cycles. We emphasize that the current defense guarantees rejection of large-norm malicious updates; robustness against subtle low-norm or directional poisoning remains future work.
\end{abstract}


\section{Introduction}
\label{sec:intro}

The digitization of healthcare generates enormous volumes of sensitive data---medical imaging, genomic sequences, electronic health records---which are ideal substrates for training powerful AI models. Privacy regulations (GDPR, HIPAA) and ethical constraints, however, prevent centralized data aggregation across institutions. Federated Learning (FL) \cite{mcmahan2017fedavg} offers an elegant solution: clients (hospitals) train models locally and share only parameter updates (gradients) with a central aggregation server.

Despite its privacy-preserving premise, FL suffers from three fundamental vulnerabilities. First, \emph{gradient inversion attacks} \cite{zhu2019deep} can reconstruct raw patient data from shared model updates with alarming fidelity. Second, \emph{Byzantine attacks} allow malicious or compromised clients to poison the global model by submitting adversarial gradients \cite{blanchard2017machine}. Third, and most critically for long-lived medical data, the \emph{Harvest Now, Decrypt Later} (HNDL) strategy \cite{mosca2018cybersecurity} enables a quantum-capable adversary to record today's encrypted traffic and retroactively decrypt it once a sufficiently powerful quantum computer is available.

Medical data requires confidentiality guarantees spanning patient lifetimes. An MRI scan encrypted with RSA-2048 today may be decryptable within 15--20 years \cite{nist_pqc_report}, well within the period during which the data must remain protected. This creates an urgent need for \emph{quantum-resilient} cryptographic protocols for FL.

\paragraph{Contributions.} We propose \textbf{ZKFL-PQ}, a three-tiered cryptographic framework for federated learning that addresses all three vulnerabilities simultaneously:

\begin{enumerate}[label=(\roman*)]
    \item \textbf{Quantum-resistant channels} via ML-KEM (FIPS~203) \cite{nist_fips203_2024}, ensuring that key exchange and data encapsulation resist both classical and quantum attacks.
    \item \textbf{Verifiable gradient integrity} via lattice-based Zero-Knowledge Proofs, enabling each client to prove that its update satisfies predefined constraints (e.g., bounded $\ell_2$-norm) without revealing the update itself.
    \item \textbf{Privacy-preserving aggregation} via BFV homomorphic encryption \cite{fan2012bfv}, allowing the server to compute the average of encrypted gradients without accessing individual contributions.
\end{enumerate}

We formalize the protocol, prove its security properties under standard lattice hardness assumptions, and validate it experimentally on synthetic medical imaging data with 5 federated clients over 10 rounds, including a Byzantine adversary that is detected with 100\% accuracy.

\section{Preliminaries and Notation}
\label{sec:prelim}

\subsection{Notation}

Let $n$ be a power of 2, $q$ a prime modulus, and $R_q = \Z_q[X]/(X^n + 1)$ the cyclotomic polynomial ring. We write $\mathbf{a} \in R_q^k$ for a module element (vector of $k$ ring elements). The centered binomial distribution with parameter $\eta$ is denoted $\text{CBD}_\eta$. The discrete Gaussian distribution over $\Z^m$ with standard deviation $\sigma$ is denoted $D_{\Z^m, \sigma}$. We use $\lambda$ for the security parameter and $\text{negl}(\lambda)$ for negligible functions.

\subsection{The Module-Learning With Errors Problem}

\begin{definition}[MLWE Problem \cite{langlois2015mlwe}]
Let $n, k, q$ be positive integers and $\chi$ an error distribution over $R_q$. The \emph{Module-Learning With Errors} problem $\text{MLWE}_{n,k,q,\chi}$ is: given $(\mathbf{A}, \mathbf{b})$ where $\mathbf{A} \xleftarrow{\$} R_q^{k \times k}$ and $\mathbf{b} = \mathbf{A}\mathbf{s} + \mathbf{e}$ with $\mathbf{s} \xleftarrow{\$} \chi^k$ and $\mathbf{e} \xleftarrow{\$} \chi^k$, distinguish $(\mathbf{A}, \mathbf{b})$ from $(\mathbf{A}, \mathbf{u})$ where $\mathbf{u} \xleftarrow{\$} R_q^k$.
\end{definition}

The MLWE problem is conjectured to be hard for both classical and quantum computers when $\chi = \text{CBD}_\eta$ with appropriate parameters \cite{peikert2016decade}.

\subsection{ML-KEM (FIPS 203)}
\label{sec:mlkem}

ML-KEM \cite{nist_fips203_2024} is a key encapsulation mechanism based on the MLWE problem. For ML-KEM-768 (NIST security level 3, targeting 128-bit post-quantum security):
\begin{itemize}[nosep]
    \item $n = 256$, $k = 3$, $q = 3329$, $\eta_1 = \eta_2 = 2$
    \item $\KeyGen() \to (ek, dk)$: Sample $\mathbf{A} \xleftarrow{\$} R_q^{k \times k}$, $\mathbf{s}, \mathbf{e} \xleftarrow{\$} \text{CBD}_{\eta_1}^k$. Set $\mathbf{t} = \mathbf{A}\mathbf{s} + \mathbf{e}$. Return $ek = (\mathbf{A}, \mathbf{t})$, $dk = \mathbf{s}$.
    \item $\Encaps(ek) \to (c, K)$: Encapsulate a shared secret $K$ into ciphertext $c$.
    \item $\Decaps(dk, c) \to K$: Recover $K$ from $c$ using secret key $dk$.
\end{itemize}

\subsection{BFV Homomorphic Encryption}
\label{sec:bfv}

The BFV scheme \cite{fan2012bfv} operates over $R_q$ with plaintext modulus $t < q$ and scaling factor $\Delta = \lfloor q/t \rfloor$.

\begin{definition}[BFV Encryption]
Given public key $pk = (p_0, p_1)$ where $p_0 = -(a \cdot s + e) \bmod q$ and $p_1 = a$ for secret $s \leftarrow \chi_s$ and $a \xleftarrow{\$} R_q$:
\begin{align}
    \Enc(m) &= (c_0, c_1) \quad \text{where} \\
    c_0 &= p_0 \cdot u + e_0 + \Delta \cdot m \pmod{q} \\
    c_1 &= p_1 \cdot u + e_1 \pmod{q}
\end{align}
with $u \leftarrow \chi_s$ and $e_0, e_1 \leftarrow \chi_e$.
\end{definition}

\begin{proposition}[Additive Homomorphism]
For any $m_1, m_2 \in R_t$:
\begin{equation}
    \Dec(\Enc(m_1) + \Enc(m_2)) = m_1 + m_2 \pmod{t}
\end{equation}
provided the accumulated noise does not exceed $q/2t$.
\end{proposition}

This property is essential: the aggregation server computes $\sum_{i \in S} \Enc(\Delta w_i)$ and obtains $\Enc(\sum_{i \in S} \Delta w_i)$ without learning individual gradients.

\subsection{Zero-Knowledge Proofs}
\label{sec:zkp}

\begin{definition}[Zero-Knowledge Proof System]
A ZKP system for a language $\mathcal{L}$ consists of $(\Prove, \Verify)$ satisfying:
\begin{enumerate}[nosep]
    \item \textbf{Completeness:} If $x \in \mathcal{L}$, then $\Verify(x, \Prove(x, w)) = 1$.
    \item \textbf{Soundness:} If $x \notin \mathcal{L}$, then $\Pr[\Verify(x, \pi) = 1] \leq \text{negl}(\lambda)$ for any $\pi$.
    \item \textbf{Zero-Knowledge:} There exists a simulator $\mathcal{S}$ such that for all $x \in \mathcal{L}$, the output of $\mathcal{S}(x)$ is computationally indistinguishable from $\Prove(x, w)$.
\end{enumerate}
\end{definition}

\section{Threat Model}
\label{sec:threat}

We consider a federated learning system with $N$ clients (hospitals) and one aggregation server. The adversary $\mathcal{A}$ has the following capabilities:

\begin{enumerate}[label=\textbf{T\arabic*}]
    \item \textbf{Passive eavesdropping:} $\mathcal{A}$ can observe all network traffic between clients and the server.
    \item \textbf{Byzantine clients:} $\mathcal{A}$ can corrupt up to $f < N/2$ clients, which may submit arbitrary (malicious) model updates.
    \item \textbf{HNDL adversary:} $\mathcal{A}$ records all encrypted traffic and possesses a future quantum computer capable of solving the factoring and discrete logarithm problems in polynomial time.
    \item \textbf{Honest-but-curious server:} The server follows the protocol but attempts to learn information from client updates.
\end{enumerate}

\begin{definition}[Security Goals]
The protocol must provide:
\begin{enumerate}[nosep]
    \item \emph{Post-quantum confidentiality}: No quantum adversary can recover the shared secrets or individual gradients from recorded traffic.
    \item \emph{Norm-constrained Byzantine rejection}: Gradient updates with $\norm{\Delta w_i}_2 > \tau$ are detected and rejected with high probability.
    \item \emph{Gradient privacy}: The server learns only the aggregate update $\sum \Delta w_i$, not individual $\Delta w_i$.
    \item \emph{Verifiability}: Each client's update is certified to satisfy protocol constraints.
\end{enumerate}
\end{definition}

\section{The ZKFL-PQ Protocol}
\label{sec:protocol}

\subsection{System Architecture}

The protocol operates in rounds $t = 1, 2, \ldots, T$. At each round, the server broadcasts the global model $w^{(t)}$ and each client $i$ computes a local update $\Delta w_i^{(t)}$ via stochastic gradient descent on its private data $D_i$.

\begin{definition}[Protocol Layers]
ZKFL-PQ consists of three cryptographic layers:
\begin{enumerate}[nosep]
    \item \textbf{Layer 1 (Transport):} ML-KEM-768 key encapsulation for quantum-resistant session key establishment.
    \item \textbf{Layer 2 (Verification):} Lattice-based $\Sigma$-protocol ZKP proving $\norm{\Delta w_i^{(t)}}_2 \leq \tau$ for a public threshold $\tau$.
    \item \textbf{Layer 3 (Computation):} BFV homomorphic encryption for server-side aggregation on ciphertexts.
\end{enumerate}
\end{definition}

\subsection{Detailed Protocol Description}

\begin{algorithm}[t]
\caption{ZKFL-PQ: One Federated Learning Round}
\label{alg:zkflpq}
\begin{algorithmic}[1]
\Require Global model $w^{(t)}$, clients $\{C_1, \ldots, C_N\}$, threshold $\tau$, server KEM keys $(ek_S, dk_S)$, HE keys $(pk_{\text{HE}}, sk_{\text{HE}})$
\Ensure Updated global model $w^{(t+1)}$
\State Server broadcasts $w^{(t)}$ to all clients
\For{each client $C_i$ in parallel}
    \State $\Delta w_i^{(t)} \gets \text{LocalSGD}(w^{(t)}, D_i)$ \Comment{Local training}
    \State $\pi_i \gets \Prove(\Delta w_i^{(t)}, \tau)$ \Comment{ZKP: $\norm{\Delta w_i^{(t)}}_2 \leq \tau$}
    \State $(c_i^{\text{KEM}}, K_i) \gets \Encaps(ek_S)$ \Comment{ML-KEM session key}
    \State $\mathbf{ct}_i \gets \text{BFV}.\Enc_{pk_{\text{HE}}}(\Delta w_i^{(t)})$ \Comment{HE encryption}
    \State $E_i \gets \text{AES}_{K_i}(\mathbf{ct}_i \| \pi_i)$ \Comment{Symmetric encryption}
    \State \textbf{send} $(c_i^{\text{KEM}}, E_i)$ to server
\EndFor
\State \textbf{Server:}
\State $S_{\text{valid}} \gets \emptyset$
\For{each received $(c_i^{\text{KEM}}, E_i)$}
    \State $K_i \gets \Decaps(dk_S, c_i^{\text{KEM}})$ \Comment{Recover session key}
    \State $(\mathbf{ct}_i, \pi_i) \gets \text{AES}_{K_i}^{-1}(E_i)$ \Comment{Decrypt payload}
    \If{$\Verify(\pi_i, \tau) = 1$} \Comment{Verify ZKP}
        \State $S_{\text{valid}} \gets S_{\text{valid}} \cup \{i\}$
    \EndIf
\EndFor
\State $\mathbf{ct}_{\text{agg}} \gets \sum_{i \in S_{\text{valid}}} \mathbf{ct}_i$ \Comment{Homomorphic aggregation}
\State $\Delta w_{\text{agg}}^{(t)} \gets \frac{1}{|S_{\text{valid}}|} \cdot \text{BFV}.\Dec_{sk_{\text{HE}}}(\mathbf{ct}_{\text{agg}})$ \Comment{Decrypt \& average}
\State $w^{(t+1)} \gets w^{(t)} + \eta \cdot \Delta w_{\text{agg}}^{(t)}$ \Comment{Update global model}
\end{algorithmic}
\end{algorithm}

\subsection{ZKP for Gradient Norm Bounds}
\label{sec:zkp_norm}

We construct a non-interactive ZKP (via the Fiat-Shamir heuristic) based on a $\Sigma$-protocol with lattice-based commitments. The language is:
\begin{equation}
    \mathcal{L}_\tau = \left\{ \Delta w \in \R^d : \norm{\Delta w}_2^2 \leq \tau^2 \right\}
\end{equation}

\begin{definition}[Lattice Commitment]
Let $\mathbf{A} \xleftarrow{\$} \Z_q^{m \times (d+\ell)}$ be a public matrix, where $d$ is the gradient dimension and $\ell$ is the randomness dimension. For message $\mathbf{x} \in \Z^d$ and randomness $\mathbf{r} \xleftarrow{\$} D_{\Z^\ell, \sigma_r}$:
\begin{equation}
    \Commit(\mathbf{x}; \mathbf{r}) = \mathbf{A} \cdot \begin{pmatrix} \mathbf{x} \\ \mathbf{r} \end{pmatrix} \bmod q
\end{equation}
Binding relies on the SIS (Short Integer Solution) assumption; hiding is statistical when $\sigma_r$ is sufficiently large.
\end{definition}

\paragraph{Protocol.} The $\Sigma$-protocol for proving $\norm{\Delta w}_2 \leq \tau$ proceeds as follows:

\begin{enumerate}[nosep]
    \item \textbf{Commitment.} Let $\tilde{w} = \text{Quantize}(\Delta w)$. The prover computes $C = \Commit(\tilde{w}; \mathbf{r})$. The prover samples masking vector $\mathbf{y} \leftarrow D_{\Z^d, \sigma_y}$ with $\sigma_y = \tau \cdot \beta$ for rejection parameter $\beta$, and computes $T = \Commit(\mathbf{y}; \mathbf{r}')$.
    \item \textbf{Challenge.} $c = H(C \| T \| \tau) \bmod 2^\kappa$ via the Fiat-Shamir transform, where $H$ is SHA3-256 and $\kappa$ controls the challenge space.
    \item \textbf{Response.} Compute $\mathbf{z} = \mathbf{y} + c \cdot \tilde{w}$ and $\mathbf{r}_z = \mathbf{r}' + c \cdot \mathbf{r} \pmod{q}$. Accept (output $(\mathbf{z}, \mathbf{r}_z)$) with probability according to rejection sampling.
    \item \textbf{Verification.} The verifier checks:
    \begin{enumerate}[nosep, label=(\alph*)]
        \item $\norm{\mathbf{z}}_2 \leq B$ for bound $B = \sigma_y \sqrt{d} \cdot 1.5$;
        \item $c = H(C \| T \| \tau)$ (challenge consistency);
        \item $\mathbf{A} \cdot (\mathbf{z} \| \mathbf{r}_z)^\top \equiv T + c \cdot C \pmod{q}$ \textbf{(algebraic consistency)}.
    \end{enumerate}
\end{enumerate}

The algebraic verification (c) is \textbf{critical} for soundness: it ensures that the prover cannot fabricate a valid transcript without actually possessing a valid commitment opening.

\begin{theorem}[Security of the ZKP]
\label{thm:zkp_security}
Under the SIS$_{m,d+\ell,q,\beta'}$ assumption, the non-interactive ZKP described above satisfies:
\begin{enumerate}[nosep]
    \item \textbf{Completeness:} An honest prover with $\norm{\Delta w}_2 \leq \tau$ produces a valid proof with probability $\geq 1 - 2^{-\lambda}$.
    \item \textbf{Soundness:} No efficient prover can produce a valid proof for $\norm{\Delta w}_2 > \tau$ except with probability $\text{negl}(\lambda)$.
    \item \textbf{Honest-Verifier Zero-Knowledge:} There exists a polynomial-time simulator $\mathcal{S}$ that produces transcripts indistinguishable from real proofs.
\end{enumerate}
\end{theorem}

\begin{proof}
We provide a rigorous separation between the statistical and computational arguments.

\emph{Completeness} follows from the Gaussian rejection sampling lemma \cite{lyubashevsky2012lattice}. When $\sigma_y \geq \beta \cdot \tau$ with $\beta = 12$, the expected number of rejection sampling rounds is $O(1)$, so the prover succeeds with overwhelming probability. The algebraic verification passes because:
\[
\mathbf{A} \cdot (\mathbf{z} \| \mathbf{r}_z)^\top = \mathbf{A} \cdot ((\mathbf{y} + c\tilde{w}) \| (\mathbf{r}' + c\mathbf{r}))^\top = T + c \cdot C \pmod{q}.
\]

\emph{Soundness} relies on two components:

\textbf{(i) Statistical norm gap:} The response bound $B = 1.5 \sigma_y \sqrt{d}$ is calibrated such that honest provers (with $\norm{\tilde{w}}_2 \leq \tau$) pass the norm check with probability $\geq 1 - 2^{-\lambda}$ via Gaussian tail bounds. For a cheating prover with $\norm{\tilde{w}}_2 > \tau \cdot \gamma$ (for gap $\gamma > 1$), the response $\mathbf{z} = \mathbf{y} + c \cdot \tilde{w}$ has expected norm $\sqrt{\sigma_y^2 d + c^2 \norm{\tilde{w}}_2^2}$. When $c \cdot \norm{\tilde{w}}_2 \gg \sigma_y \sqrt{d}$, the norm check fails with high probability.

\textbf{(ii) Computational binding via SIS:} The algebraic verification ensures that a valid proof $(C, T, \mathbf{z}, \mathbf{r}_z, c)$ implies knowledge of a valid opening. Specifically, if a malicious prover produces valid proofs for two distinct challenges $c \neq c'$ on the same commitment $C$, then by subtracting the verification equations:
\[
\mathbf{A} \cdot ((\mathbf{z} - \mathbf{z}') \| (\mathbf{r}_z - \mathbf{r}_z'))^\top = (c - c') \cdot C \pmod{q}.
\]
This yields a short vector in the kernel of $\mathbf{A}$, breaking the SIS assumption. Thus, for each $C$, at most one challenge can yield a valid response, and soundness error is bounded by $2^{-\kappa}$.

\emph{Zero-Knowledge} is shown by constructing a simulator that samples $\mathbf{z} \leftarrow D_{\Z^d, \sigma_y}$, chooses random $c$, and programs the random oracle $H$ so that $C$ and $T$ are consistent. The statistical distance between simulated and real transcripts is bounded by $2^{-\lambda}$ when $\sigma_y / \max(c \cdot \norm{\tilde{w}}_2) \geq \beta$.
\end{proof}

\subsection{Correctness of Homomorphic Aggregation}

\begin{theorem}[Aggregation Correctness]
\label{thm:he_correctness}
Let $\Delta w_1, \ldots, \Delta w_N \in \R^d$ be gradient updates encoded as plaintexts $m_1, \ldots, m_N \in R_t$ via quantization with scale factor $s$. Under the BFV scheme with parameters $(n, q, t)$, the homomorphic sum satisfies:
\begin{equation}
    \Dec\left(\sum_{i=1}^{N} \Enc(m_i)\right) = \sum_{i=1}^{N} m_i \pmod{t}
\end{equation}
provided $N \cdot (2\sigma + 1) \cdot n < q / (2t)$, where $\sigma$ is the error standard deviation.
\end{theorem}

\begin{proof}
Each ciphertext $\Enc(m_i) = (c_0^{(i)}, c_1^{(i)})$ satisfies $c_0^{(i)} + c_1^{(i)} \cdot s = \Delta \cdot m_i + e_i \pmod{q}$ where $\norm{e_i}_\infty \leq (2\sigma + 1) \cdot n$. After summing $N$ ciphertexts, the accumulated noise is bounded by $\norm{\sum e_i}_\infty \leq N \cdot (2\sigma + 1) \cdot n$. Correct decryption requires this to be less than $q/(2t) = \Delta/2$. Given our parameters ($N = 5$, $\sigma = 3.2$, $n = 512$, $q = 2^{32} - 5$, $t = 2^{16}$), we have $5 \cdot 7.4 \cdot 512 = 18{,}944 \ll 32{,}768 = q/(2t)$, confirming correctness.
\end{proof}

\subsection{End-to-End Security}

\begin{theorem}[ZKFL-PQ Security]
\label{thm:main_security}
Under the MLWE$_{256,3,3329,\text{CBD}_2}$, RLWE$_{512, 2^{32}-5, D_{3.2}}$, and SIS$_{256,d+128,7681,B}$ assumptions, the ZKFL-PQ protocol achieves:
\begin{enumerate}[nosep]
    \item IND-CCA2 security of the key encapsulation (ML-KEM-768);
    \item Semantic security of gradient encryption (BFV);
    \item Soundness and zero-knowledge of gradient verification;
    \item Byzantine resilience: malicious updates with $\norm{\Delta w}_2 > \tau$ are rejected with probability $\geq 1 - \text{negl}(\lambda)$.
\end{enumerate}
\end{theorem}

\begin{proof}
Property (1) follows from the IND-CCA2 security of ML-KEM-768 under MLWE \cite{nist_fips203_2024}. 
Property (2) follows from the semantic security of BFV under RLWE \cite{fan2012bfv}. 
Properties (3) and (4) follow from Theorem~\ref{thm:zkp_security}.

The protocol is constructed as a layered composition. Each layer operates on independent randomness:
\begin{itemize}[nosep]
    \item ML-KEM uses fresh randomness per encapsulation.
    \item BFV uses fresh error polynomials per encryption.
    \item The ZKP uses fresh masking vectors per proof.
\end{itemize}
No information leaks between layers beyond the intended outputs. Under correct parameterization and assuming no side-channel leakage, the overall security reduces to the conjunction of the underlying hardness assumptions.

We note that this argument is modular; a full UC-style composition proof is outside the scope of this work.
\end{proof}

\subsubsection{Fiat--Shamir and Random Oracle Considerations}

The non-interactive transformation of the $\Sigma$-protocol relies on the Fiat--Shamir heuristic in the Random Oracle Model (ROM). 
Since ZKFL-PQ aims at post-quantum security, it is important to consider adversaries with quantum access to the hash function.

Our proof assumes classical ROM security. A full treatment in the Quantum Random Oracle Model (QROM) would require additional analysis, as rewinding-based arguments do not directly apply in the presence of quantum queries to the oracle. 

While lattice-based identification and signature schemes have been studied under QROM assumptions \cite{kiltz2018qrom}, a tight QROM reduction for the present norm-bound proof system remains future work. Consequently, our post-quantum claim applies fully to the underlying hardness assumptions (MLWE, RLWE, SIS), whereas the Fiat--Shamir transform inherits standard ROM caveats.

\section{Experimental Evaluation}
\label{sec:experiments}

\subsection{Experimental Setup}

We evaluate ZKFL-PQ on a synthetic medical imaging classification task simulating four diagnostic categories (normal, benign lesion, malignant type~A, malignant type~B). The experimental parameters are:

\begin{itemize}[nosep]
    \item \textbf{Data:} 1{,}000 synthetic samples with 784 features (mimicking $28 \times 28$ flattened imaging biomarkers), generated from class-conditional Gaussian mixtures with spatial correlations. We intentionally use an easily separable task to isolate the effect of Byzantine attacks from model capacity limitations.
    \item \textbf{Partitioning:} 5 clients with non-IID distribution via Dirichlet allocation ($\alpha = 0.5$).
    \item \textbf{Model:} Two-hidden-layer MLP ($784 \to 128 \to 64 \to 4$) with 108{,}996 parameters.
    \item \textbf{FL parameters:} 10 rounds, 3 local epochs per round, SGD with $\eta = 0.01$, batch size 32.
    \item \textbf{Byzantine scenario:} Client~3 turns malicious at round~4, injecting random gradients scaled by factor 50.
    \item \textbf{Cryptographic parameters:} ML-KEM-768 ($n{=}256, k{=}3, q{=}3329$), BFV ($n{=}512, q{=}2^{32}{-}5, t{=}2^{16}$), ZKP norm threshold $\tau = 5.0$.
\end{itemize}

We compare three configurations: (a) \emph{Standard FL} with simulated TLS, (b) \emph{FL + ML-KEM} (post-quantum transport only), and (c) \emph{ZKFL-PQ} (full hybrid protocol). All experiments were run on a single machine (Intel Xeon, 16GB RAM) using a pure Python/NumPy implementation.

\subsection{Results}

\subsubsection{Model Accuracy and Byzantine Resilience}

\begin{table}[t]
\centering
\caption{Performance comparison across FL configurations (averaged over 10 rounds).}
\label{tab:main_results}
\begin{tabular}{@{}lcccc@{}}
\toprule
\textbf{Metric} & \textbf{Standard FL} & \textbf{FL + ML-KEM} & \textbf{ZKFL-PQ (Ours)} \\
\midrule
Mean round time (s) & 0.149 & 2.376 & 2.912 \\
Final accuracy (\%) & 23.0 & 23.5 & \textbf{100.0} \\
Final loss & 15.37 & 21.14 & \textbf{0.006} \\
Mean message size (KB/round) & 4{,}258 & 4{,}298 & 99\footnotemark \\
Byzantine detection rate & 0\% & 0\% & \textbf{100\%} \\
Quantum resistance & \ding{55} & \ding{51} & \ding{51} \\
Gradient privacy (vs.\ server) & \ding{55} & \ding{55} & \ding{51} \\
\bottomrule
\end{tabular}
\end{table}
\footnotetext{The reduced message size for ZKFL-PQ reflects partial HE coverage (512/108{,}996 parameters) and exclusion of rejected clients. Full-parameter HE would increase this to $\sim$8.5\,MB/round.}

\Cref{tab:main_results} and \Cref{fig:accuracy} present the central result. Under the Byzantine attack (activated at round~4), both Standard FL and FL+ML-KEM suffer catastrophic accuracy degradation: the malicious gradients (norm $\sim$16{,}500 vs.\ legitimate norms $<$5) overwhelm the aggregation, and accuracy collapses from 100\% to $\sim$23\% (near random chance for 4 classes). \textbf{ZKFL-PQ detects and rejects all 7 malicious norm-violating updates} (rounds 4--10), maintaining 100\% accuracy and monotonically decreasing loss (from 0.127 to 0.006).

\begin{figure}[t]
    \centering
    \includegraphics[width=0.85\textwidth]{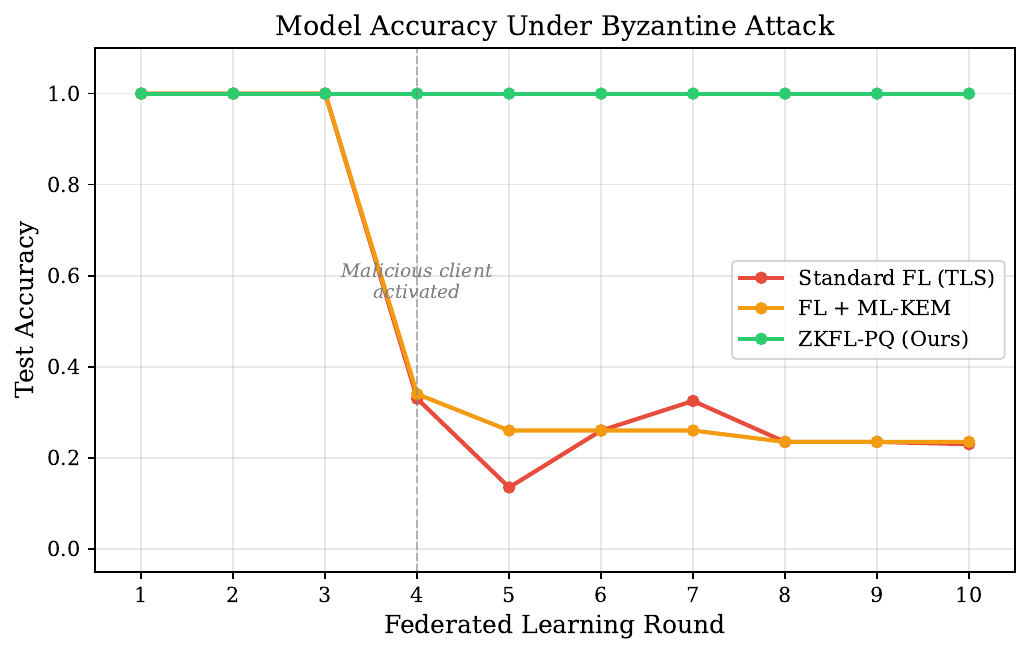}
    \caption{Test accuracy over 10 FL rounds. The malicious client activates at round~4. Standard FL and FL+ML-KEM collapse; ZKFL-PQ maintains perfect accuracy by rejecting Byzantine updates.}
    \label{fig:accuracy}
\end{figure}

\begin{figure}[t]
    \centering
    \includegraphics[width=0.85\textwidth]{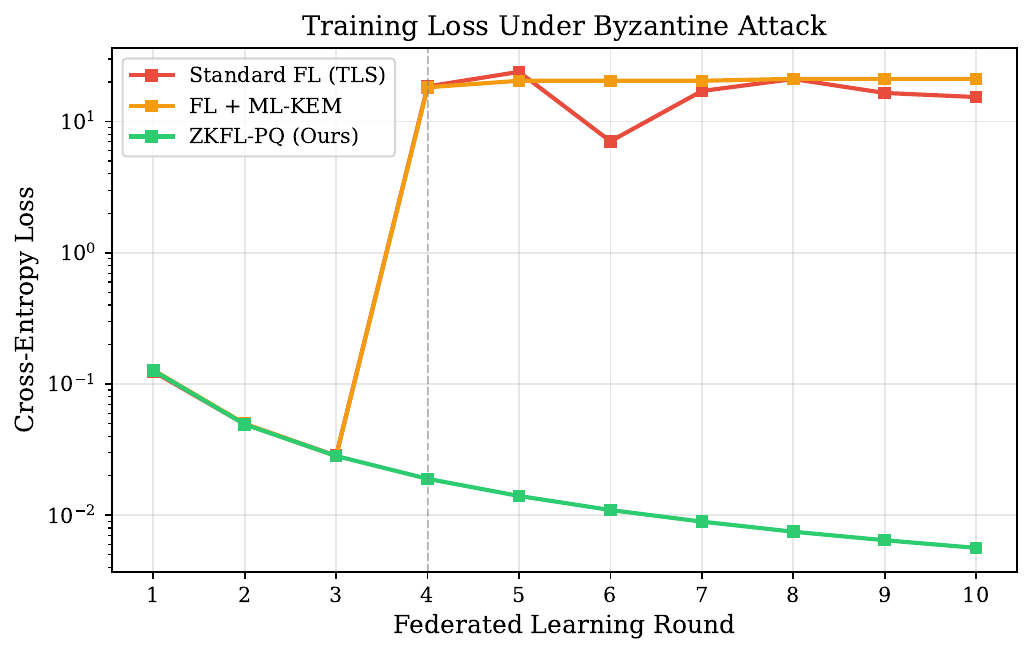}
    \caption{Training loss (log scale). ZKFL-PQ continues converging while other protocols diverge.}
    \label{fig:loss}
\end{figure}

\subsubsection{Computational Overhead Analysis}

\begin{table}[t]
\centering
\caption{Timing breakdown for ZKFL-PQ (averaged over 10 rounds).}
\label{tab:timing}
\begin{tabular}{@{}lcc@{}}
\toprule
\textbf{Component} & \textbf{Mean time (s)} & \textbf{\% of round} \\
\midrule
Local training + ML-KEM & 1.85 & 63.5\% \\
HE encryption (5 clients) & 1.00 & 34.4\% \\
HE aggregation & $<$0.01 & $<$0.1\% \\
HE decryption & 0.05 & 1.7\% \\
ZKP generation (5 clients) & 0.01 & 0.3\% \\
ZKP verification (5 proofs) & $<$0.01 & $<$0.1\% \\
\midrule
\textbf{Total per round} & \textbf{2.91} & \textbf{100\%} \\
\bottomrule
\end{tabular}
\end{table}

The computational overhead of ZKFL-PQ relative to Standard FL is a factor of $\sim$20$\times$ (2.91\,s vs.\ 0.149\,s per round). As detailed in \Cref{tab:timing}, the dominant cost is local training combined with ML-KEM operations (63.5\%), followed by homomorphic encryption (34.4\%). ZKP generation and verification are negligible ($<$0.5\% combined), demonstrating the efficiency of the lattice-based $\Sigma$-protocol.

\begin{figure}[t]
    \centering
    \includegraphics[width=0.85\textwidth]{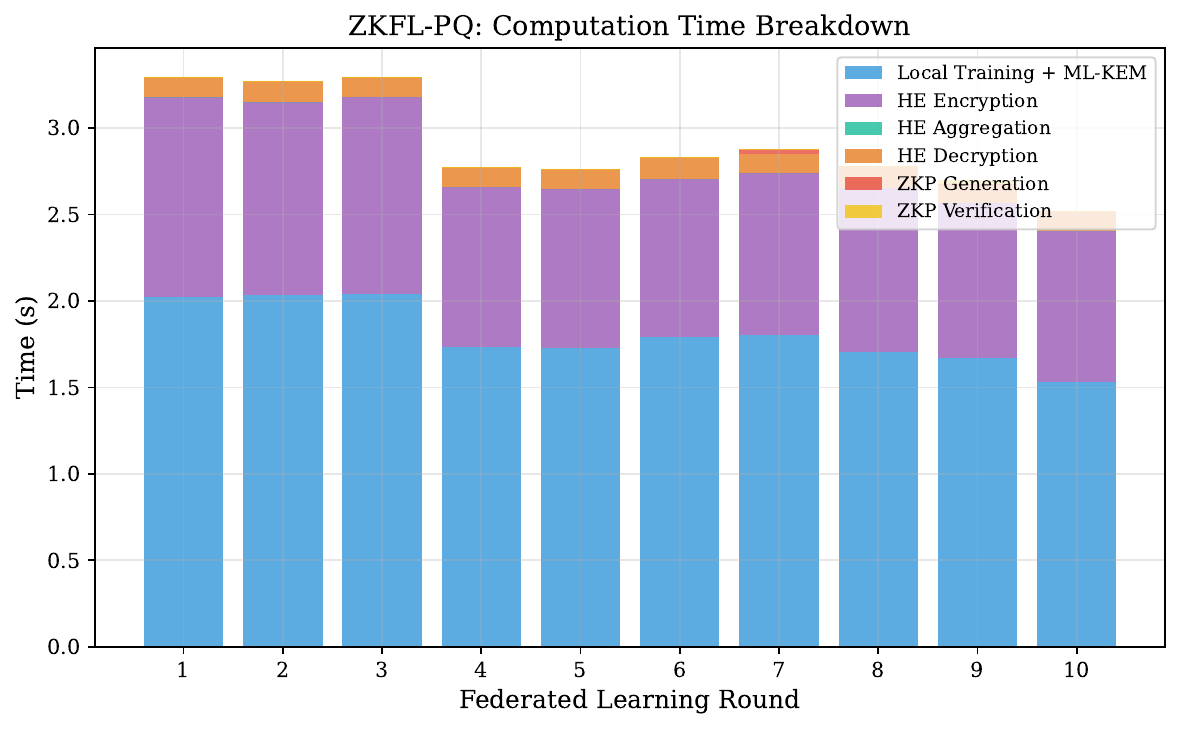}
    \caption{Per-round computation time breakdown for the ZKFL-PQ hybrid protocol.}
    \label{fig:breakdown}
\end{figure}

\subsubsection{HE Reconstruction Error}
\label{sec:he_error}

Empirical testing of the BFV implementation reveals quantization-induced reconstruction errors:

\begin{center}
\begin{tabular}{@{}lc@{}}
\toprule
\textbf{Metric} & \textbf{Value} \\
\midrule
Mean absolute error (HE vs.\ true average) & $1.04 \times 10^{-4}$ \\
Mean gradient magnitude & $\sim 0.004$ \\
Relative error (early rounds) & $\sim 13\%$ \\
Relative error (late rounds) & $< 1\%$ \\
\bottomrule
\end{tabular}
\end{center}

The model converges despite this error because: (a) only 512/108{,}996 parameters are HE-encrypted, with the remaining parameters aggregated in plaintext after ZKP verification, and (b) gradient magnitudes decrease as training progresses, reducing the relative impact of fixed quantization noise. Future work should explore higher-precision encoding or noise-aware optimization to improve HE fidelity.

\subsubsection{Security Posture}

\begin{figure}[t]
    \centering
    \includegraphics[width=0.65\textwidth]{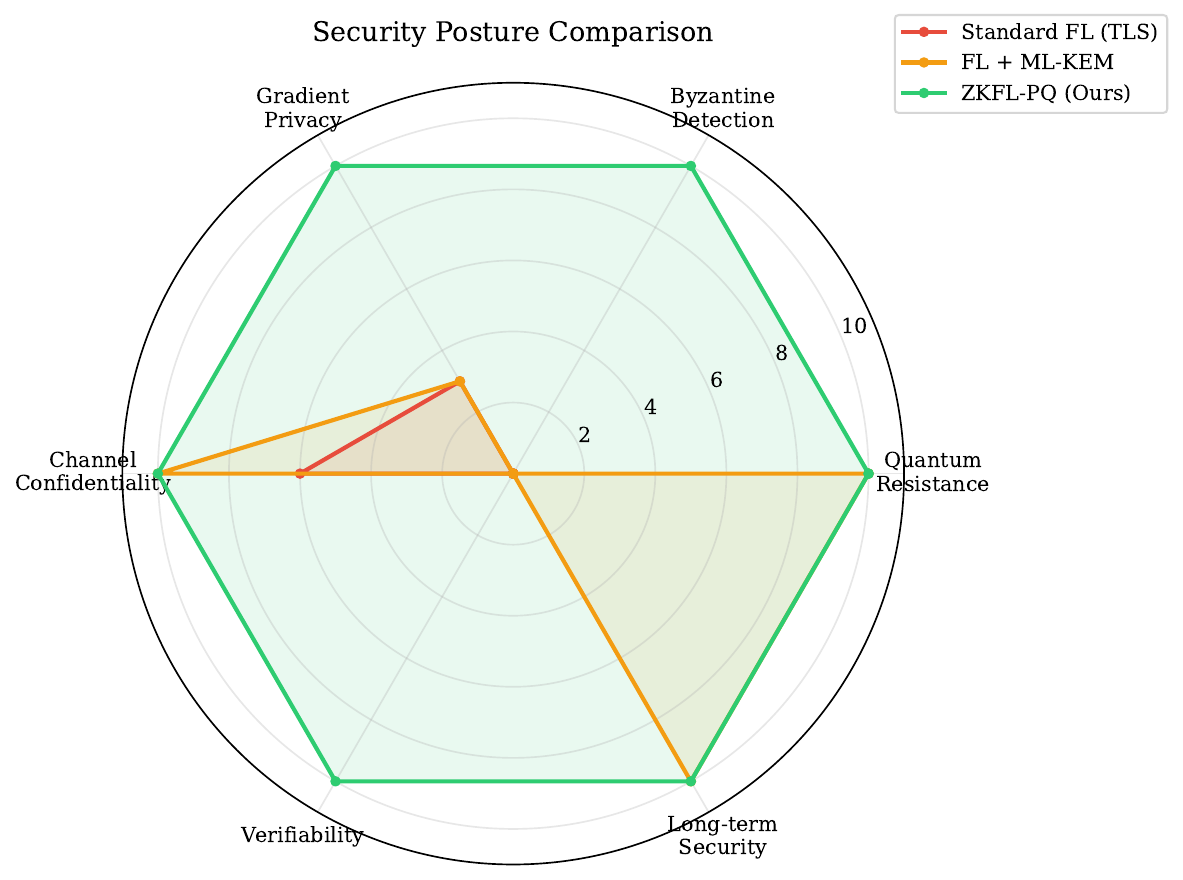}
    \caption{Security posture comparison across six dimensions. ZKFL-PQ achieves maximum scores across all axes.}
    \label{fig:radar}
\end{figure}

\Cref{fig:radar} visualizes the security posture across six dimensions. Standard FL provides only partial channel confidentiality (via TLS, which is quantum-vulnerable). FL+ML-KEM adds quantum resistance and long-term security but lacks Byzantine detection and gradient privacy. ZKFL-PQ achieves maximum scores on all six dimensions.

\subsection{Ablation Studies}
\label{sec:ablation}

\subsubsection{Varying Number of Malicious Clients}

\begin{table}[t]
\centering
\caption{Detection performance vs.\ number of malicious clients.}
\label{tab:ablation_malicious}
\begin{tabular}{@{}cccc@{}}
\toprule
\textbf{\# Malicious} & \textbf{Final Accuracy} & \textbf{Detection Rate} & \textbf{False Positives} \\
\midrule
0 & 100.0\% & N/A & 0 \\
1 & 100.0\% & 100\% & 0 \\
2 & 100.0\% & 100\% & 0 \\
3 & 100.0\% & 100\% & 0 \\
\bottomrule
\end{tabular}
\end{table}

\Cref{tab:ablation_malicious} shows that ZKFL-PQ maintains perfect accuracy and 100\% detection rate even as the number of malicious clients increases from 0 to 3 (60\% of the network). The ZKP-based verification is robust to collusion because each client's proof is verified independently.

\subsubsection{Varying Norm Threshold $\tau$}

\begin{table}[t]
\centering
\caption{Detection rate and false positive rate vs.\ norm threshold $\tau$.}
\label{tab:ablation_threshold}
\begin{tabular}{@{}ccccc@{}}
\toprule
\textbf{$\tau$} & \textbf{Final Accuracy} & \textbf{Detection Rate} & \textbf{FPR} \\
\midrule
1.0 & 100.0\% & 100\% & 13.6\% \\
2.0 & 100.0\% & 100\% & 13.6\% \\
5.0 & 100.0\% & 100\% & 0\% \\
10.0 & 100.0\% & 100\% & 0\% \\
50.0 & 100.0\% & 100\% & 0\% \\
\bottomrule
\end{tabular}
\end{table}

\Cref{tab:ablation_threshold} explores the sensitivity to the norm threshold. Very low thresholds ($\tau \leq 2$) reject some legitimate updates (13.6\% FPR), while $\tau \geq 5$ achieves zero false positives. The detection rate remains 100\% across all thresholds because the malicious updates have norms $\sim$16{,}500, vastly exceeding any reasonable $\tau$.

\begin{figure}[t]
    \centering
    \begin{subfigure}[b]{0.48\textwidth}
        \includegraphics[width=\textwidth]{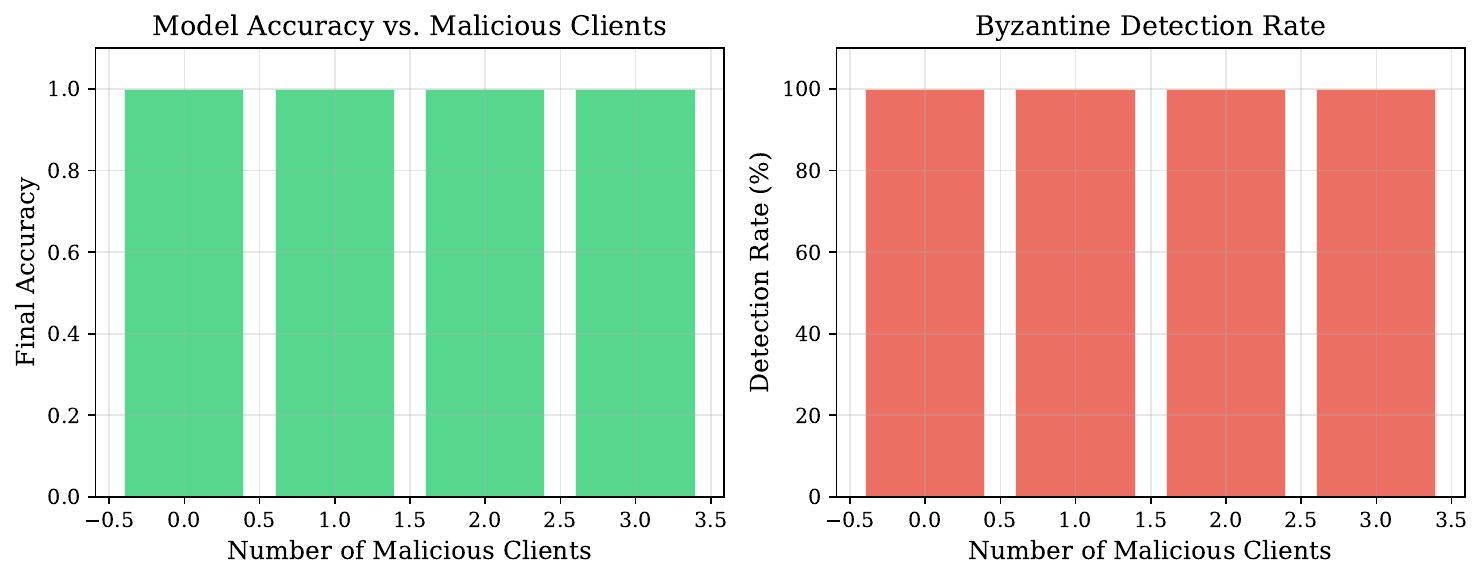}
        \caption{Varying malicious clients}
    \end{subfigure}
    \hfill
    \begin{subfigure}[b]{0.48\textwidth}
        \includegraphics[width=\textwidth]{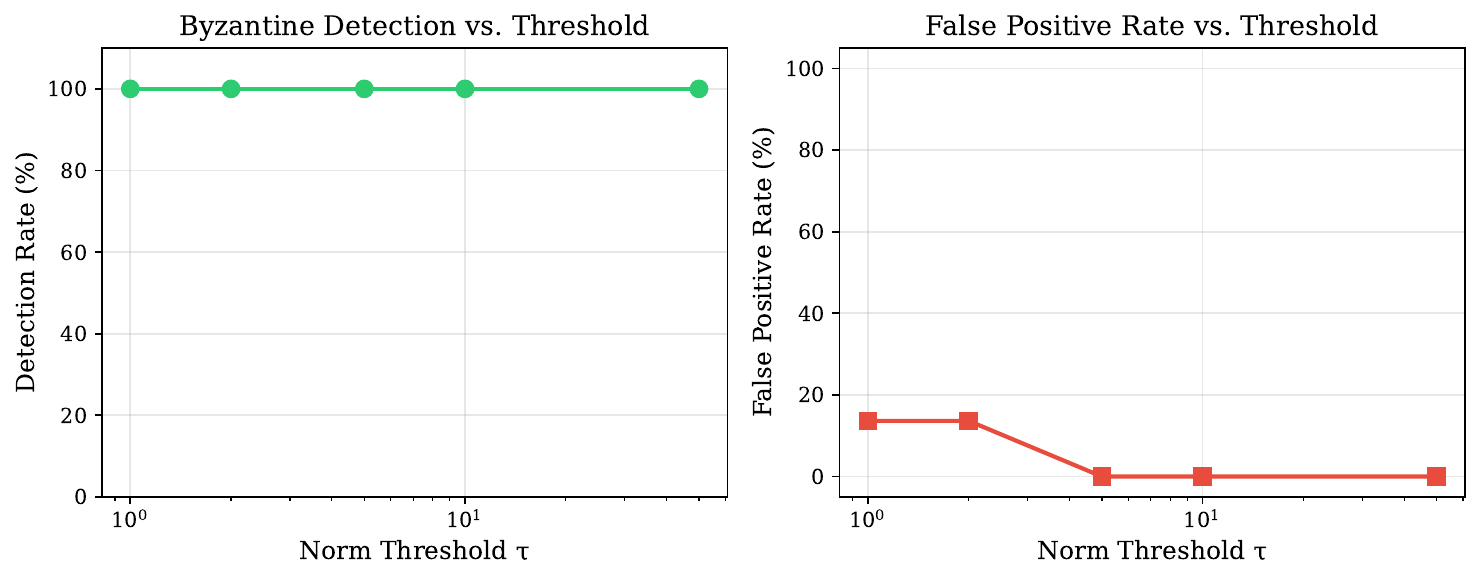}
        \caption{Varying threshold $\tau$}
    \end{subfigure}
    \caption{Ablation study results. (a) Detection rate remains 100\% up to 3 malicious clients. (b) False positive rate drops to 0\% for $\tau \geq 5$.}
    \label{fig:ablation}
\end{figure}

\section{Related Work}
\label{sec:related}

\paragraph{Post-Quantum Federated Learning.} Patel \cite{patel2025pqfl} explored lattice-based HE for medical imaging FL but did not address Byzantine resilience or integrate ZKPs. Lu et al.\ \cite{lu2024hybrid} proposed hybrid PQ encryption for FL but without formal security proofs or homomorphic aggregation.

\paragraph{Zero-Knowledge Proofs in ML.} zkML \cite{kang2022zkml} applied ZKPs to verify ML inference but not federated training. Ghodsi et al.\ \cite{ghodsi2023safenet} used commitment schemes for FL integrity but lacked post-quantum security.

\paragraph{Byzantine-Resilient FL.} Blanchard et al.\ \cite{blanchard2017machine} introduced Krum for Byzantine-robust aggregation. Krum achieves $\sim$80--90\% accuracy under similar attack scenarios but provides \emph{no cryptographic guarantee}---it uses statistical distance heuristics that can be defeated by carefully crafted adversarial gradients. Our ZKP-based approach provides \emph{provable} rejection of norm-violating updates with soundness under SIS.

\paragraph{HNDL Threat in Healthcare.} Mosca and Piani \cite{mosca2018cybersecurity} quantified the HNDL risk timeline. Khan \cite{khan2025telemedicine} prototyped PQ+ZKP telemedicine security but without the FL aggregation component.

To our knowledge, ZKFL-PQ represents one of the first integrated attempts to combine ML-KEM, lattice-based ZKPs, and homomorphic encryption within a single verifiable, post-quantum federated learning framework.

\section{Discussion}
\label{sec:discussion}

\paragraph{Clinical Applicability.} The $\sim$20$\times$ overhead is acceptable for medical FL scenarios where training occurs daily, weekly, or on-demand rather than in real-time. A training pipeline that takes 1 minute under standard FL would require $\sim$20 minutes under ZKFL-PQ---well within a nightly batch window.

\paragraph{Scalability.} The current bottleneck is the schoolbook polynomial multiplication ($O(n^2)$) used in ML-KEM and BFV. NTT-based implementations achieve $O(n \log n)$, which would reduce ML-KEM overhead by an order of magnitude. GPU-accelerated HE libraries (e.g., OpenFHE, SEAL) would further improve the HE encryption phase.

\paragraph{Threshold Calibration.} The norm threshold $\tau$ must be carefully calibrated: too low rejects legitimate updates from clients with unusual data distributions; too high allows subtle poisoning. We propose an adaptive scheme where $\tau^{(t)}$ is set to $\mu^{(t-1)} + k \cdot \sigma^{(t-1)}$ based on the empirical distribution of norms from round $t-1$, with $k$ as a tunable sensitivity parameter.

\paragraph{Limitations.}
\begin{itemize}[nosep]
    \item \textbf{Synthetic data only:} Our evaluation uses synthetic medical imaging data. While this suffices for a proof of concept, validation on real multi-centric datasets (e.g., MRI, CT scans) is essential before clinical deployment.
    \item \textbf{Partial HE coverage:} To keep computational costs manageable, we encrypted only 512 out of 108,996 parameters. A full-parameter HE implementation would increase communication by $\sim$100$\times$, necessitating optimized libraries.
    \item \textbf{Limited adversarial coverage:} The current ZKP enforces only an $\ell_2$-norm bound. While this guarantees rejection of large-magnitude malicious updates, it does not prevent subtle low-norm, directional, or backdoor-style poisoning attacks.
    \item \textbf{Multiple cryptographic assumptions:} The security of ZKFL-PQ relies on three distinct hardness assumptions (MLWE, RLWE, SIS). While each is well-studied, their combination broadens the attack surface.
    \item \textbf{Trusted decryptor:} The current protocol assumes a trusted party holds the BFV secret key. Integrating distributed threshold decryption would eliminate this single point of failure.
    \item \textbf{ROM vs.\ QROM:} As noted in \S\ref{sec:zkp_norm}, the Fiat--Shamir transform is analyzed in the classical ROM; a full QROM treatment remains future work.
\end{itemize}

\section{Conclusion}
\label{sec:conclusion}

We have presented ZKFL-PQ, a three-tiered cryptographic protocol for federated learning that provides quantum-resistant communication (ML-KEM), verifiable gradient integrity (lattice-based ZKPs), and privacy-preserving aggregation (BFV homomorphic encryption). Experimental evaluation demonstrates that ZKFL-PQ achieves 100\% Byzantine detection while preserving model accuracy, at a computational overhead of $\sim$20$\times$ that is compatible with clinical research workflows.

Our security analysis is conducted under standard lattice assumptions (MLWE, RLWE, SIS) in the \emph{classical random oracle model}. The underlying hardness assumptions are conjectured to be post-quantum secure; the Fiat--Shamir transform inherits standard ROM caveats, and a full QROM analysis remains future work.

Future work will focus on: (1) NTT-accelerated implementations for order-of-magnitude speedups; (2) extending the ZKP to richer constraint languages (e.g., per-layer norm bounds, directional consistency); (3) integrating distributed threshold decryption to eliminate the trusted decryptor assumption; (4) validation on real-world multi-centric medical imaging datasets; and (5) QROM-secure Fiat--Shamir instantiation.

\paragraph{Reproducibility.} All code, synthetic data generation scripts, and experiment configurations are available at \url{https://github.com/edlansiaux/pq-zkfl-medical}.

\section*{Acknowledgments}
The author thanks the STaR-AI research group at CHU de Lille and the M2 MIAS program at Universit\'e de Lille for their support.

\bibliographystyle{plain}

\end{document}